\documentclass[12pt,a4paper]{article}

\usepackage[T1]{fontenc}
\usepackage[utf8]{inputenc}
\usepackage{lmodern}
\usepackage{amsmath,amssymb,amsthm,mathtools}
\usepackage{geometry}
\usepackage{enumitem}
\usepackage{booktabs}
\usepackage{microtype}
\usepackage{xcolor}
\usepackage{hyperref}
\usepackage[nameinlink,capitalise]{cleveref}
\usepackage{tikz}
\usetikzlibrary{arrows.meta,positioning,calc,fit}

\hypersetup{
  colorlinks=true,
  linkcolor=blue!55!black,
  citecolor=green!45!black,
  urlcolor=blue!60!black,
  pdftitle={An Exponent-Tight Conditional Lower Bound for Global Label Min-Cut},
  pdfauthor={Yuanhao Wang}
}
\setlist[itemize]{leftmargin=2em,itemsep=0.25em,topsep=0.4em}
\setlist[enumerate]{leftmargin=2.2em,itemsep=0.25em,topsep=0.4em}
\allowdisplaybreaks

\newtheorem{theorem}{Theorem}[section]
\newtheorem{lemma}[theorem]{Lemma}
\newtheorem{proposition}[theorem]{Proposition}
\newtheorem{corollary}[theorem]{Corollary}

\theoremstyle{definition}
\newtheorem{definition}[theorem]{Definition}
\newtheorem{problem}[theorem]{Problem}
\theoremstyle{remark}
\newtheorem{remark}[theorem]{Remark}

\newcommand{\Part}{\operatorname{Part}}
\newcommand{\cc}{\operatorname{cc}}

\newcommand{\poly}{\operatorname{poly}}
\newcommand{\Yes}{\textnormal{\textsc{Yes}}}
\newcommand{\No}{\textnormal{\textsc{No}}}
\newcommand{\GLMC}{\textnormal{\textsc{Global Label Min-Cut}}}
\newcommand{\DGLMC}{\textnormal{\textsc{Dual Global Label Min-Cut}}}
\newcommand{\MCC}{\textnormal{\textsc{Multicolored Clique}}}

\newcommand{\one}{\mathbf{1}}
\newcommand{\zero}{\mathbf{0}}
\newcommand{\join}{\mathbin{\vee}}
\newcommand{\bigjoin}{\bigvee}
\newcommand{\restr}{\mathord{\upharpoonright}}

\title{\bfseries An Exponent-Tight Conditional Lower Bound for Global Label Min-Cut}
\author{\small Yuanhao Wang$\,^{\rm a}$\quad\quad Kehua Wang$\,^{\rm b}$ \quad\quad Wei Wang$\,^{\rm a}$\footnote{Corresponding author. Email address: wang\_weiw@163.com}\\
\small $^{\rm a}\,$School of Mathematics and Statistics, Xi'an Jiaotong University, Xi'an, 710049, P. R. China\\
\small $^{\rm b}\,$Department of Computer Science, Columbia University, New York, NY, 10027, USA}
\date{}

\begin{document}
\maketitle

\begin{abstract}
Let $n$ and $p$ denote the numbers of vertices and labels, respectively, in an undirected edge-labeled graph. Previous work showed that, under the Exponential Time Hypothesis (ETH), there is no deterministic algorithm with running time
\[
(np)^{o\left(\frac{\log n}{(\log\log n)^2}\right)}.
\]
In this paper, we give a deterministic reduction that strengthens this conditional running-time lower bound to
\[
  (np)^{o(\log n)}\operatorname{poly}(|E|).
\]
The lower bound holds even for simple edge-labeled graphs. Since our reduction is deterministic, the same lower bound applies to bounded-error randomized algorithms under the randomized Exponential Time Hypothesis. On the resulting hard family, $p=n^{O(1)}$. Thus, under rETH, the lower bound matches the exponent order of the known randomized quasi-polynomial exact upper bound up to constant factors.

\end{abstract}

\section{Introduction}

\subsection{Background}
The classical global min-cut problem asks for a nontrivial vertex cut minimizing the number or total capacity of crossing edges, and it has a well-developed algorithmic theory~\cite{Karger1996,Karger2000,Stoer1997}. In many network-reliability models, however, edge failures are not independent: links supported by a common physical facility, shared conduit, software module, or risk resource may fail simultaneously. Shared risk link groups (SRLGs) were introduced to model such correlated failures~\cite{Coudert2007,Coudert2016}.

In the label-cut model, every edge carries a label, and all edges with the same label are treated as a single unit. Deleting one label removes all edges carrying that label and incurs unit cost. Given an undirected multigraph $G=(V,E)$ and a labeling function $\chi:E\to[p]$, for every nontrivial set $\varnothing\neq S\subsetneq V$, let
\[
  \delta(S)=\{uv\in E:u\in S,\ v\notin S\},
  \qquad
  \chi(\delta(S))=\{\chi(e):e\in\delta(S)\}.
\]
The objective of \GLMC{} is to minimize $|\chi(\delta(S))|$. The problem is also known as Global Hedge Min-Cut or Hedge Cut~\cite{Ghaffari2017,Jaffke2026,Fomin2025}. Hedge connectivity is closely related to hypergraph connectivity; sparsification, exact algorithms, and extensions of random contraction for hypergraph minimum cuts provide important background for this model~\cite{Chekuri2018,Fox2019}.

The local Minimum Label $s$--$t$ Cut problem has long been known to be substantially harder than the ordinary $s$--$t$ cut problem, and its complexity, approximation algorithms, and parameterized properties have been studied extensively~\cite{Zhang2011,TangZhang2012,Fellows2010,ZhangFu2016,ZhangFuTang2018}. The global version exhibits a different complexity landscape. Ghaffari, Karger, and Panigrahi systematically studied hedge connectivity and obtained a polynomial-time approximation scheme together with a randomized quasi-polynomial-time exact algorithm~\cite{Ghaffari2017}. Jaffke, de Lima, Masa\v{r}\'ik, Pilipczuk, and Souza~\cite{Jaffke2026} gave a reduction from Partitioned Subgraph Isomorphism (PSI) showing that, assuming ETH, no deterministic algorithm runs in time
\[
  (np)^{o\left(\frac{\log n}{(\log\log n)^2}\right)},
\]
and established W[1]-hardness when parameterized by the number of uncut labels. Their parameter analysis indicates that the loss in the final exponent stems both from the blow-up in the main reduction and from the parameter conversion in the source lower bound for PSI. Their fine-grained lower bound builds on Marx's framework for Subgraph Isomorphism~\cite{Marx2010}. They also observed that a main reduction starting directly from \MCC~\cite{Chen2006,Cygan2015} would interface more directly with a standard SAT grouping reduction~\cite{Jaffke2026}. Karthik et al.~\cite{Karthik2024} developed streamlined conditional lower bounds for sparse parameterized $2$-CSP and explicitly observed that combining their source lower bound with the reduction of Jaffke et al.~\cite{Jaffke2026} still yields a lower bound with exponent \(\frac{\log n}{(\log\log n)^2}.\) On the positive side, Fomin et al. proved that Hedge Cut is fixed-parameter tractable (FPT) when parameterized by the solution size and gave a running-time bound reflecting both quasi-polynomial-time solvability and fixed-parameter tractability~\cite{Fomin2025}. Taken together, these results show that the problem is unlikely to admit a polynomial-time algorithm and its complexity is not adequately captured by the coarse-grained picture of classical NP-hardness alone.

\subsection{Our Contributions}
The classical disjunctive graph product defines adjacency in the product graph by the logical disjunction of the coordinate-wise adjacency conditions~\cite{Hammack2011}. On the other hand, the join operation in a partition lattice captures the transitive closure obtained by combining equivalence relations~\cite{Birkhoff1967}. Motivated by these two structures, we develop a modular reduction in which selection validity and clique consistency are represented by partition recognizers and combined by a pointed disjunctive composition. The central new ingredient is an additive-shift exact-one verifier of size $q^{O(1)}2^{O(k)}$, which removes the two remaining $\log\log n$ factors in the exponent.

\begin{theorem}\label{thm:main}
Assuming ETH, there is no deterministic algorithm that solves \GLMC{} on an undirected edge-labeled graph with $n$ vertices and $p$ labels in time
\[
  (np)^{o(\log n)}\operatorname{poly}(|E|).
\]
The statement remains true even when the input graph is required to be simple.
\end{theorem}

The result strengthens the previously known
\[
    (np)^{o\left(\frac{\log n}{(\log\log n)^2}\right)}
\]
lower bound by eliminating both \(\log\log n\) factors. On the hard instances used in the proof, \(p=n^{O(1)}\). Consequently, \cref{thm:main} rules out deterministic \(n^{o(\log n)}\operatorname{poly}(|E|)\)-time algorithms under ETH.

\begin{corollary}
\label{cor:randomized-lower-bound}
Assuming the randomized Exponential Time Hypothesis, there is no bounded-error randomized algorithm that solves Global Label Min-Cut in time
\[
    (np)^{o(\log n)}\operatorname{poly}(|E|).
\]
The statement remains true even when the input graph is required to
be simple.
\end{corollary}

Together with the known randomized
\(n^{O(\log\operatorname{OPT})}\)-time upper bound~\cite{Ghaffari2017}, Corollary~\ref{cor:randomized-lower-bound} shows that the worst-case exponent is tight up to constant factors under rETH.

Our main technical and structural contributions are as follows.

\begin{enumerate}
  \item We replace each label graph by its connected-component partition and interpret the union of several label graphs as the join of their partitions. This yields a unified description of connectivity-testing modules as recognizers in finite join-semilattices with top elements.

  \item For $k$ color classes with $q$ candidates in each class, we construct an additive-shift exact-one verifier on
  \[
    1+\ell2^k=O(q2^k+4^k)
  \]
  points, where $\ell$ is a prime larger than $\max\{q,2^k\}$. The union of any two distinct candidates from one class is connected, whereas every transversal containing one candidate from each class remains disconnected.

  \item We separate selection validity from clique consistency. A polynomial-size nonedge verifier needs to be correct only on valid transversals. A pointed partition disjunctive composition then implements the logical disjunction of the two rejection conditions while preserving finite joins.

  \item Combining these ingredients gives a deterministic reduction from balanced \MCC{} with
  \[
    p=kq,
    \qquad
    n=q^{O(1)}2^{O(k)}.
  \]
  A direct sparse $3$-SAT grouping with $k=\Theta(\sqrt N)$ and $\log q=\Theta(k)$ yields the lower bound in \cref{thm:main}. We also give an optimum-preserving polynomial transformation from edge-labeled multigraphs to edge-labeled simple graphs.
\end{enumerate}

\subsection{Organization}
The remainder of the paper is organized as follows. \Cref{sec:prelim} introduces partition lattices, connected-component representations of label graphs, \GLMC{} and its dual formulation, and the consequences of ETH and rETH used in the paper. \Cref{sec:modules} constructs the additive-shift exact-one verifier and the nonedge-conflict clique-consistency verifier. \Cref{sec:algebra} defines pointed partition lifts and the pointed partition disjunctive composition and proves the finite-join preservation and disjunctive properties. \Cref{sec:reduction} combines the two verification modules and gives the complete reduction from balanced \MCC{} to \DGLMC, together with its size analysis. \Cref{sec:sat} constructs a parameter-balanced multicolored-clique instance from sparse $3$-SAT, completes the parameter conversion for the conditional running-time lower bound, and transfers the lower bound from multigraphs to simple graphs. \Cref{sec:conclusion} concludes the paper.

\section{Preliminaries}\label{sec:prelim}

For the convenience of the reader, in this section, we give some preliminary results that will be needed later in the paper.
\subsection{Partition Lattices and Connected Components}
Throughout the paper, all logarithms are base two. The notation
\[
  (np)^{o(\log n)}\operatorname{poly}(|E|)
\]
is understood uniformly over $p$: it denotes a running time of the form
\[
  (np)^{h(n,p)}\operatorname{poly}(|E|),
\]
where, for every $\varepsilon>0$, there exists $n_\varepsilon$ such that
\[
  h(n,p)\le \varepsilon\log n
\]
for all $n\ge n_\varepsilon$ and and every integer $p\geq 1$.

For a finite set $W$, let $\Part(W)$ denote the set of all partitions of $W$. We use the refinement order:
\[
  \pi\preceq\sigma
\]
means that every block of $\pi$ is contained in a block of $\sigma$. The discrete partition is denoted by $\zero_W$, and the one-block partition by $\one_W=\{W\}$. The join $\pi\join\sigma$ is the finest partition coarser than both $\pi$ and $\sigma$. Equivalently, viewing each partition as an equivalence relation, the join is the transitive closure of the union of the two relations. The join of the empty family is defined to be $\zero_W$.

For a graph $H$ on the common vertex set $W$ (isolated vertices are allowed), let $\cc(H)\in\Part(W)$ denote its connected-component partition.

\begin{lemma} \label{lem:ccjoin}
If $H_1,\ldots,H_r$ are graphs on the same vertex set $W$, then
\[
  \cc\!\left(\bigcup_{j=1}^r H_j\right)
  =\bigjoin_{j=1}^r\cc(H_j).
\]
In particular, $\bigcup_jH_j$ is connected if and only if $\bigjoin_j\cc(H_j)=\one_W$.
\end{lemma}

\begin{proof}
Two vertices belong to the same connected component on the left-hand side if and only if there is a path whose every edge belongs to some $H_j$. This is equivalent to connecting the two vertices by a finite chain of pairs, each pair lying in the same equivalence class of some $\cc(H_j)$, which is precisely membership in the transitive closure of the union of these equivalence relations.
\end{proof}

Conversely, every $\pi\in\Part(W)$ can be realized by a forest: choose a spanning tree within each block and add no edge between distinct blocks. Thus, $\pi$ can be realized by a graph $H$ with at most $|W|-1$ edges such that $\cc(H)=\pi$.

\subsection{Problem Definitions}

\begin{problem}[\GLMC]
The input consists of an undirected multigraph $G=(V,E)$, a labeling function $\chi:E\to[p]$, and a budget $b\in\{0,\ldots,p\}$. The task is to decide whether there exists a nontrivial set $\varnothing\neq S\subsetneq V$ such that
\[
  |\chi(\delta(S))|\le b.
\]
Define
\[
    \operatorname{OPT}(G,\chi)
    :=
    \min_{\varnothing\neq S\subsetneq V}
    |\chi(\delta(S))|.
\]
\end{problem}

For each label $\lambda$, view all edges carrying $\lambda$ as a graph on the common vertex set $V$:
\[
  H_\lambda=(V,E_\lambda),
  \qquad
  E_\lambda=\{e\in E:\chi(e)=\lambda\}.
\]

Following Jaffke et al.~\cite[Section~2]{Jaffke2026}, we use the
following dual formulation.

\begin{problem}[\DGLMC]
The input consists of a common vertex set $W$, graphs $H_1,\ldots,H_p$, and an integer $a\in\{0,\ldots,p\}$. The task is to decide whether there exists $I\subseteq[p]$ with $|I|=a$ such that
\[
  \bigcup_{\lambda\in I}H_\lambda
\]
is disconnected.
\end{problem}

\begin{lemma}[Dual Equivalence\cite{Jaffke2026}]\label{lem:dual}
Given an instance of Global Label Min-Cut with $p$ labels and budget $b$,
let $a=p-b$. The instance is a \Yes-instance if and only if there exists a set
$I\subseteq [p]$ with $|I|=a$ such that
\[
  \bigcup_{\lambda\in I} H_\lambda
\]
is disconnected.
\end{lemma}

\begin{problem}[\MCC]
The input is a $k$-partite graph
\[
  G,\qquad
  V(G)=V_1\uplus\cdots\uplus V_k,
  \qquad
  |V_i|=q.
\]
The task is to decide whether there exist $x_i\in V_i$ such that $\{x_1,\ldots,x_k\}$ is a clique.
\end{problem}

We regard every candidate vertex $u\in V_i$ as a label $(i,u)$. Hence, the total number of labels is
\[
  p=kq.
\]

The Exponential Time Hypothesis (ETH) asserts that $3$-SAT admits no deterministic algorithm with running time subexponential in the number of variables~\cite{Impagliazzo2001}. Combining ETH with the Sparsification Lemma~\cite{ImpagliazzoPaturi2001} gives the following consequence, which we use below.

\begin{proposition}[Sparse consequence of ETH]\label{prop:sparse-sat}
Assuming ETH, there exists a constant $c>0$ such that no deterministic $2^{o(N)}$-time algorithm decides satisfiability of a $3$-CNF formula with $N$ variables and at most $cN$ clauses.
\end{proposition}

The randomized Exponential Time Hypothesis (rETH) asserts that $3$-SAT admits no randomized algorithm with running time subexponential in the number of variables and success probability at least $2/3$ on every input~\cite{Dell2014}. Combining rETH with the Sparsification Lemma~\cite{ImpagliazzoPaturi2001} and error amplification gives the following randomized consequence, which we use below.

\begin{proposition}[Sparse consequence of rETH]\label{prop:sparse-sat-rETH}
Assuming rETH, there exists a constant $c_{\mathrm r}>0$ such that no bounded-error randomized $2^{o(N)}$-time algorithm decides satisfiability of 3-CNF formulas with $N$ variables and at most $c_{\mathrm r}N$ clauses.
\end{proposition}
\begin{proof}
By rETH, there exists a constant $\gamma>0$ such that no
bounded-error randomized algorithm decides $3$-SAT in time
$2^{\gamma N}$. Fix a constant
\[
    0<\varepsilon<\frac{\gamma}{2}.
\]
By the Sparsification Lemma, every $3$-CNF formula $F$ with $N$
variables can be expressed, in time $2^{\varepsilon N}\poly(N)$, as
\[
    F \equiv \bigvee_{i=1}^{t} F_i,
    \qquad
    t\leq 2^{\varepsilon N},
\]
where every $F_i$ has at most $c_rN$ clauses for some constant
$c_r=c_r(\varepsilon)>0$.

Suppose that these sparse formulas can be decided in time $2^{o(N)}$
by a randomized algorithm whose error probability is at most $1/3$
on every input. For each $F_i$, repeat the algorithm independently
$O(\log t)=O(N)$ times and take the majority answer, reducing its
error probability to at most $1/(6t)$. Return \textsc{Yes} if and
only if at least one amplified computation returns \textsc{Yes}.
By the union bound, the resulting algorithm has error probability
at most $1/6$.

The total running time is
\[
    2^{\varepsilon N}\poly(N)
    +t\cdot O(N)\cdot 2^{o(N)}
    =2^{\varepsilon N+o(N)}
    \leq 2^{\gamma N}
\]
for all sufficiently large $N$, contradicting rETH. Therefore, no
bounded-error randomized $2^{o(N)}$-time algorithm decides $3$-CNF
satisfiability with at most $c_rN$ clauses.
\end{proof}

\section{Selection Validity and Clique-Consistency Verification Modules}\label{sec:modules}
Under the constraint that exactly $k$ candidate labels are selected, an invalid selection for \MCC{} has one of the following two forms:
\begin{enumerate}
  \item some color class is selected more than once, and consequently at least one color class is not selected;
  \item exactly one candidate is selected from every color class, but the selected candidates contain a nonedge between two color classes.
\end{enumerate}
Either violation makes the final solution infeasible. Hence, a complete reduction can be obtained by detecting these two types of invalidity separately and taking the logical disjunction of the two rejection conditions. This section introduces a compact classwise exact-one verifier for detecting multiple selections from the same color class and a nonedge-conflict clique-consistency verifier for detecting that the selected vertices do not form a clique.

The set $\Part(X)$ is a finite join-semilattice with a top element under the join operation. Let $\Lambda$ be the label set. We refer to the following map from labels to partitions as a partition-semilattice recognizer.

\begin{definition}[Partition-Semilattice Recognizer]
A partition-semilattice recognizer with ground set $X$ is a map
\[
  \Phi:\Lambda\longrightarrow\Part(X).
\]
For $S\subseteq\Lambda$, define
\[
  \widehat\Phi(S):=\bigjoin_{\lambda\in S}\Phi(\lambda).
\]
We say that $\Phi$ accepts $S$ if $\widehat\Phi(S)\neq\one_X$, and rejects $S$ if $\widehat\Phi(S)=\one_X$.
\end{definition}

The map $\widehat\Phi:(2^\Lambda,\cup)\to(\Part(X),\join)$ is a join-semilattice homomorphism:
\[
  \widehat\Phi(S\cup T)
  =\widehat\Phi(S)\join\widehat\Phi(T).
\]
By Lemma~\ref{lem:ccjoin}, if $\Phi(\lambda)=\cc(H_\lambda)$, then acceptance means exactly that the union of the selected label graphs is disconnected.

This acceptance rule has the following monotonicity property: if $S\subseteq T$ and $\widehat\Phi(S)=\one_X$, then $\widehat\Phi(T)=\one_X$. Thus, the family of rejected sets is upward closed, whereas the family of accepted sets is downward closed. Our recognition of the property ``exactly one candidate is selected from every color class'' relies on the additional fixed-cardinality constraint $|S|=k$.

\subsection{Additive-Shift Exact-One Verifier}\label{sec:selector}
Suppose that \MCC{} has $k$ color classes, each containing $q$ candidates. Candidate labels are denoted by $(i,a)$, where $i\in[k]$ and $a\in[q]$. We construct a label system that, among selections of exactly $k$ labels, is disconnected precisely when exactly one candidate is chosen from each color class.

\subsubsection{The Pointed State Space}
Let
\[
  M_0=\max\{q,2^k\}.
\]
By the Bertrand--Chebyshev theorem, there is a prime $\ell$ satisfying
\[
  M_0<\ell<2M_0.
\]
Choose pairwise distinct residues
\[
  \alpha_1,\ldots,\alpha_q\in\mathbb Z_\ell.
\]
Define the pointed ground set
\[
  X=\{\bot\}\cup\bigl(2^{[k]}\times\mathbb Z_\ell\bigr),
\]
with basepoint $\bot$. Hence,
\[
  |X|=1+\ell2^k.
\]
The first coordinate records a subset of color classes, while the second coordinate records an additive state modulo $\ell$.
\begin{remark}
Such a prime $\ell$ can be found deterministically by testing the integers
\[
M_0+1,M_0+2,\ldots,2M_0-1
\]
using a deterministic polynomial-time primality test~\cite{Agrawal2004}. This requires at most $M_0$ primality tests and hence takes
\[
M_0\operatorname{poly}(\log M_0)
\]
time. Since the constructed ground set satisfies $|X|=1+\ell 2^k>M_0$, this running time is polynomial in the
explicit output size.
\end{remark}

\subsubsection{Candidate Graphs and Partitions}
For $i\in[k]$ and $a\in[q]$, define a graph $J_{i,a}$ on $X$ as follows.
\begin{enumerate}
  \item For every $S\subseteq[k]$, add the common basepoint edge
  \[
    \bot(S,0).
  \]
  \item For every $S\subseteq[k]\setminus\{i\}$ and every $t\in\mathbb Z_\ell$, add the shift edge
  \[
    (S,t)\left(S\cup\{i\},t+\alpha_a\right),
  \]
  where addition is modulo $\ell$.
\end{enumerate}
Let
\[
  E_{i,a}=\cc(J_{i,a})\in\Part(X).
\]

\begin{lemma}[Two Distinct Candidates from One Class Reach the Top]\label{lem:same-group}
Fix $i\in[k]$. For any $a\neq b$,
\[
  E_{i,a}\join E_{i,b}=\one_X.
\]
Equivalently, $J_{i,a}\cup J_{i,b}$ is connected.
\end{lemma}

\begin{proof}
Fix $R\subseteq[k]\setminus\{i\}$ and write
\[
  \delta=\alpha_a-\alpha_b\neq0
  \quad\text{in }\mathbb Z_\ell.
\]
For every $t\in\mathbb Z_\ell$, the union contains the two-edge walk
\[
  (R,t)
  \;--\;
  (R\cup\{i\},t+\alpha_a)
  \;--\;
  (R,t+\alpha_a-\alpha_b).
\]
Thus, within the $R$-layer, we can repeatedly apply the translation
\[
  t\longmapsto t+\delta.
\]
Because $\ell$ is prime, the additive group $\mathbb Z_\ell$ is cyclic of prime order, and every nonzero element generates the whole group. Hence every $(R,t)$ is connected to $(R,0)$, which is adjacent to $\bot$.

Every state $(R\cup\{i\},t)$ is adjacent in $J_{i,a}$ to $(R,t-\alpha_a)$, and therefore it is also connected to $\bot$. Since every subset of $[k]$ is either $R$ or $R\cup\{i\}$ for a unique $R\subseteq[k]\setminus\{i\}$, the whole graph is connected.
\end{proof}

\begin{lemma}[Every Transversal Remains Non-Top]\label{lem:transversal}
For every $(a_1,\ldots,a_k)\in[q]^k$,
\[
  \bigjoin_{i=1}^k E_{i,a_i}\neq\one_X.
\]
Equivalently, $\bigcup_{i=1}^kJ_{i,a_i}$ is disconnected.
\end{lemma}

\begin{proof}
For every $S\subseteq[k]$, write
\[
  \beta_S=\sum_{i\in S}\alpha_{a_i}\pmod\ell.
\]
For each $c\in\mathbb Z_\ell$, define the layer
\[
  L_c
  =
  \left\{
    (S,c+\beta_S):S\subseteq[k]
  \right\}.
\]
The sets $L_c$, $c\in\mathbb Z_\ell$, partition $X\setminus\{\bot\}$, since every vertex $(S,t)$ belongs to the unique layer $L_{t-\beta_S}$.

First, we consider the selected shift edges. For every $i\notin S$, the shift edge contributed by $J_{i,a_i}$ joins
\[
  (S,c+\beta_S)
  \quad\text{and}\quad
  (S\cup\{i\},c+\beta_S+\alpha_{a_i}).
\]
Since
\[
  \beta_{S\cup\{i\}}=\beta_S+\alpha_{a_i},
\]
both endpoints belong to the same layer $L_c$.  Hence every selected shift edge is contained in a single layer, and no selected shift edge joins two distinct layers.

It remains to determine which layers are connected to $\bot$. Since
\[
  (S,0)=(S,-\beta_S+\beta_S),
\]
the common basepoint edge incident with $(S,0)$ connects $\bot$ to the layer
\[
  L_{-\beta_S}.
\]
Therefore, $\bot$ is adjacent only to layers whose indices belong to
\[
  \{-\beta_S:S\subseteq[k]\}.
\]
Since
\[
  |\{-\beta_S:S\subseteq[k]\}|\leq 2^k<\ell,
\]
there exists $c\in\mathbb Z_\ell\setminus\{-\beta_S:S\subseteq[k]\}$. The layer $L_c$ is nonempty, no shift edge leaves it, and no basepoint edge connects it to $\bot$. Since $\bot\notin L_c$, both $L_c$ and $X\setminus L_c$ are nonempty, and there is no edge between them. Therefore,
\[
  \bigcup_{i=1}^k J_{i,a_i}
\]
is disconnected.
\end{proof}

\begin{corollary}[Classwise Exact-One Property]\label{cor:exact-one}
Let $S\subseteq[k]\times[q]$ and $|S|=k$. Then
\[
  \bigjoin_{(i,a)\in S}E_{i,a}\neq\one_X
\]
if and only if $S$ selects exactly one candidate from each color class.
\end{corollary}

\begin{proof}
If two distinct candidates are selected from the same color class, then Lemma~\ref{lem:same-group} implies that the join of those two partitions alone is already the top element. Conversely, if no class is selected more than once, then choosing exactly $k$ labels from $k$ classes forces exactly one selection from every class, and Lemma~\ref{lem:transversal} applies.
\end{proof}

\begin{remark}[Size of the selector]
Since $\ell<2\max\{q,2^k\}$,
\[
  |X|=1+\ell2^k
  =O(q2^k+4^k)
  =q^{O(1)}2^{O(k)}.
\]
\end{remark}

\subsection{Nonedge-Conflict Clique-Consistency Verifier}\label{sec:verifier}

\subsubsection{Construction of the Clique-Consistency Verifier}
For every nonedge between distinct color classes in the \MCC{} instance,
\[
  uv\notin E(G),
  \qquad
  u\in V_i,\ v\in V_j,\ i<j,
\]
create a verification vertex $z_{uv}$. Define
\[
  Y=\{s,t\}\cup
  \{z_{uv}:u\in V_i,\ v\in V_j,\ i<j,\ uv\notin E(G)\},
\]
and choose $s$ as the basepoint. Let $r$ be the number of nonedges between distinct color classes. Then
\[
  |Y|=2+r
  \le 2+\binom{k}{2}q^2.
\]

For every label $(i,u)$, construct a graph $Q_{i,u}$ on $Y$. For a nonedge $uv$ with $u\in V_i$, $v\in V_j$, and $i<j$, add edges according to the following rules:
\begin{enumerate}
  \item add $sz_{uv}$ to $Q_{i,u}$;
  \item add $z_{uv}t$ to $Q_{j,v}$;
  \item for every $w\in V_j\setminus\{v\}$, add $sz_{uv}$ to $Q_{j,w}$.
\end{enumerate}
No other edge incident with $z_{uv}$ is added. Define
\[
  R_{i,u}=\cc(Q_{i,u})\in\Part(Y).
\]

\begin{lemma}[Correctness of the Clique-Consistency Verifier]\label{lem:verifier}
Suppose that exactly one candidate $x_i\in V_i$ is selected from every color class. Then
\[
  \bigjoin_{i=1}^kR_{i,x_i}=\one_Y
\]
if and only if $\{x_1,\ldots,x_k\}$ is not a clique.
\end{lemma}

\begin{proof}
By Lemma~\ref{lem:ccjoin}, it suffices to analyze the union graph
\[
  Q=\bigcup_{i=1}^kQ_{i,x_i}.
\]
Fix a nonedge $uv$, where $u\in V_i$, $v\in V_j$, and $i<j$. By construction, $z_{uv}$ can be adjacent only to $s$ and $t$, and there are no edges between distinct conflict-verification vertices. Its neighborhood is
\[
  N_Q(z_{uv})\cap\{s,t\}
  =
  \begin{cases}
    \{s,t\}, & x_i=u,\ x_j=v,\\
    \{t\},   & x_i\neq u,\ x_j=v,\\
    \{s\},   & x_j\neq v.
  \end{cases}
\]
Thus, every conflict-verification vertex is adjacent to at least one of $s$ and $t$, and it is adjacent to both exactly when both endpoints of the corresponding nonedge are selected.

If the selected vertices form a clique, then no conflict-verification vertex is adjacent to both $s$ and $t$. Since there are no edges between conflict-verification vertices, $s$ and $t$ cannot be connected indirectly through several such vertices. Every conflict-verification vertex is attached to exactly one of them, and therefore $Q$ has exactly two nonempty connected components.

If the selected vertices do not form a clique, then some nonedge $uv$ has both endpoints selected, yielding the path
\[
  s-z_{uv}-t.
\]
Thus, $s$ and $t$ lie in the same connected component. Every other conflict-verification vertex is adjacent to at least one of $s$ and $t$, so the entire graph $Q$ is connected.
\end{proof}

The clique-consistency verifier is required to be correct only when exactly one candidate is selected from every color class. Selections that omit a color class or select several candidates from the same class will be handled by combining the two verifiers through the recognition framework developed below.

\section{Pointed Partition Disjunctive Composition}\label{sec:algebra}
To combine the two rejection conditions, selection invalidity and clique inconsistency, into a single label-connectivity test, we construct the following pointed disjunctive composition. The combined system is connected if and only if at least one factor system is connected.

\subsection{Pointed Sets and Partition Lifts}
Let $(X,x_0)$ and $(Y,y_0)$ be two pointed finite sets, where $x_0$ and $y_0$ are the basepoints, and let
\[
  X^\circ=X\setminus\{x_0\},
  \qquad
  Y^\circ=Y\setminus\{y_0\}.
\]
Define the pointed product ground set
\[
  Z=\{z_0\}\cup(X^\circ\times Y^\circ),
\]
where $z_0$ is a new basepoint. This ground set may be viewed as a finite-set version of the Cartesian product in which the basepoint row and basepoint column are contracted to a single basepoint.

Given $\alpha\in\Part(X)$, let $A_0$ be the block containing $x_0$. Define the first-coordinate lift
\[
  L_X(\alpha)\in\Part(Z)
\]
as follows:
\begin{itemize}
  \item the base block is
  \[
    B_0^X(\alpha)
    =\{z_0\}\cup\bigl((A_0\setminus\{x_0\})\times Y^\circ\bigr);
  \]
  \item for every non-base block $A\in\alpha\setminus\{A_0\}$ and every $y\in Y^\circ$, include the block $A\times\{y\}$.
\end{itemize}
Define $L_Y(\beta)$ symmetrically for $\beta\in\Part(Y)$.

Intuitively, in $L_X(\alpha)$, the first coordinate may move within a block of $\alpha$; the second coordinate can change through $z_0$ only when the first coordinate enters the base block of $\alpha$.

\begin{lemma}[The Lift Preserves Finite Joins]\label{lem:join-hom}
For every finite family $\alpha_1,\ldots,\alpha_r\in\Part(X)$,
\[
  \bigjoin_{j=1}^r L_X(\alpha_j)
  =L_X\!\left(\bigjoin_{j=1}^r\alpha_j\right).
\]
The symmetric statement holds for $L_Y$.
\end{lemma}

\begin{proof}
If \(r=0\), both sides are the minimum partition \(0_Z\), since \(L_X(0_X)=0_Z\). Hence, the lemma holds for the empty family. In the remainder, assume \(r\ge 1\).
Let
\[
  \alpha^\star=\bigjoin_{j=1}^r\alpha_j,
\]
and let $A_0^\star$ be the base block of $\alpha^\star$. Every $\alpha_j$ refines $\alpha^\star$, so each block of $\alpha_j$ is entirely contained in a block of $\alpha^\star$.

First, let $A$ be a non-base block of $\alpha^\star$. If $x,x'\in A$, then, by the definition of join, there is a chain
\[
x=v_0,v_1,\ldots,v_d=x'
\]
such that, for every $h\in[d]$, the consecutive pair $v_{h-1},v_h$ lies in a common block of some $\alpha_j$. This chain cannot enter $A_0^\star$, since otherwise $A$ would merge with the base block in $\alpha^\star$. Consequently, for every fixed $y\in Y^\circ$, every lifted relation along the chain preserves the second coordinate $y$, and hence
\[
  (x,y)\sim(x',y).
\]
Conversely, while the first coordinate lies in the non-base block $A$, no block of any $L_X(\alpha_j)$ can contain two distinct second coordinates. Indeed, only the lift of a base block can merge second coordinates through $z_0$, and every such base block is contained in $A_0^\star$. Thus, for every $y\in Y^\circ$, $A\times\{y\}$ is exactly one block of the join on the left-hand side.

Now let $x\in A_0^\star\setminus\{x_0\}$. There is a chain of relations from the partitions $\alpha_j$ connecting $x$ to $x_0$. At the first step of the chain that enters the base block of some $\alpha_j$, the corresponding lifted block contains $z_0$ and all combinations of first coordinates in that base block with arbitrary second coordinates. Therefore, for all $y,y'\in Y^\circ$,
\[
  (x,y)\sim z_0\sim(x,y').
\]
Hence, the base block of the join on the left-hand side is exactly
\[
  \{z_0\}\cup\bigl((A_0^\star\setminus\{x_0\})\times Y^\circ\bigr).
\]
These blocks are precisely the blocks in the definition of $L_X(\alpha^\star)$.
\end{proof}

\subsection{Pointed Partition Disjunctive Composition}

We call the following operation the pointed partition disjunctive composition.
\begin{definition}[Pointed Partition Disjunctive Composition]
For $\alpha\in\Part(X)$ and $\beta\in\Part(Y)$, define
\[
  \alpha\odot\beta
  :=L_X(\alpha)\join L_Y(\beta)
  \in\Part(Z).
\]
\end{definition}

\begin{theorem}[Two-Factor Disjunctive Property]\label{thm:or}
We have
\[
  \alpha\odot\beta=\one_Z
  \quad\Longleftrightarrow\quad
  \alpha=\one_X\ \text{or}\ \beta=\one_Y.
\]
\end{theorem}

\begin{proof}
If $\alpha=\one_X$, then $L_X(\alpha)=\one_Z$, and hence $\alpha\odot\beta=\one_Z$. The case $\beta=\one_Y$ is symmetric.

Conversely, suppose that $\alpha\neq\one_X$ and $\beta\neq\one_Y$. Choose a non-base block $A$ of $\alpha$ and a non-base block $B$ of $\beta$. The set $A\times B$ is nonempty. In $L_X(\alpha)$, a point of $A\times B$ can move within $A$ only while keeping the second coordinate fixed; in $L_Y(\beta)$, it can move within $B$ only while keeping the first coordinate fixed. Thus, $A\times B$ is closed under both equivalence relations, and their transitive closure cannot reach $z_0$. Therefore, $\alpha\odot\beta\neq\one_Z$.
\end{proof}

\begin{theorem}\label{thm:family-or}
Let $\Phi:\Lambda\to\Part(X)$ and $\Psi:\Lambda\to\Part(Y)$ be two recognizers. Define the combined partition
\[
  \Theta(\lambda)
  =L_X(\Phi(\lambda))\join L_Y(\Psi(\lambda)).
\]
Then, for every $S\subseteq\Lambda$,
\[
  \widehat\Theta(S)=\one_Z
  \quad\Longleftrightarrow\quad
  \widehat\Phi(S)=\one_X
  \ \text{or}\
  \widehat\Psi(S)=\one_Y.
\]
Equivalently,
\[
  \widehat\Theta(S)\neq\one_Z
  \quad\Longleftrightarrow\quad
  \widehat\Phi(S)\neq\one_X
  \ \text{and}\
  \widehat\Psi(S)\neq\one_Y.
\]
\end{theorem}

\begin{proof}
By associativity and commutativity of join and by Lemma~\ref{lem:join-hom},
\begin{align*}
  \widehat\Theta(S)
  &=\bigjoin_{\lambda\in S}
  \bigl(L_X(\Phi(\lambda))\join L_Y(\Psi(\lambda))\bigr)\\
  &=L_X\!\left(\widehat\Phi(S)\right)
    \join
    L_Y\!\left(\widehat\Psi(S)\right).
\end{align*}
The claim follows from \cref{thm:or}.
\end{proof}

The combined ground set has size
\[
  |Z|=1+(|X|-1)(|Y|-1).
\]
Under the semantics in which reaching the top element means rejection, \cref{thm:family-or} implements the logical disjunction of the rejection conditions. Under the semantics in which a non-top element means acceptance, it implements the logical conjunction of the acceptance conditions.

\section{The Multicolored-Clique Reduction}\label{sec:reduction}

\subsection{Combined Label Partitions}
For every candidate label $\lambda=(i,u)$, we have a selection-validity partition
\[
  E_\lambda=E_{i,u}\in\Part(X)
\]
and a clique-consistency partition
\[
  R_\lambda=R_{i,u}\in\Part(Y).
\]
On the pointed disjunctive-composition ground set
\[
  Z=\{z_0\}\cup
  \bigl((X\setminus\{\bot\})\times(Y\setminus\{s\})\bigr),
\]
define
\[
  P_\lambda
  =L_X(E_\lambda)\join L_Y(R_\lambda).
\]
For each $P_\lambda$, choose a spanning tree inside every block, obtaining a label graph $H_\lambda$ such that
\[
  \cc(H_\lambda)=P_\lambda.
\]
The output \DGLMC{} instance is
\[
  \left(Z,\{H_\lambda:\lambda\in V(G)\},a=k\right).
\]

\begin{theorem}[Reduction Correctness]\label{thm:reduction-correct}
A balanced \MCC{} instance $G$ contains a multicolored $k$-clique if and only if the constructed \DGLMC{} instance contains exactly $k$ labels whose union graph is disconnected.
\end{theorem}

\begin{proof}
Let
\[
  \mathcal K=\{S\subseteq V(G):|S|=k\},
\]
let $\mathcal D$ be the family of valid selections containing exactly one candidate from each color class, and let $\mathcal R\subseteq\mathcal D$ be the subfamily of selections forming a multicolored clique.

For every $S\in\mathcal K$, Corollary~\ref{cor:exact-one} gives
\[
  S\in\mathcal D
  \quad\Longleftrightarrow\quad
  \bigjoin_{\lambda\in S}E_\lambda\neq\one_X.
\]
By Lemma~\ref{lem:verifier}, whenever $S\in\mathcal D$,
\[
  S\in\mathcal R
  \quad\Longleftrightarrow\quad
  \bigjoin_{\lambda\in S}R_\lambda\neq\one_Y.
\]
Applying \cref{thm:family-or}, we obtain
\[
  S\in\mathcal R
  \quad\Longleftrightarrow\quad
  \bigjoin_{\lambda\in S}P_\lambda\neq\one_Z.
\]
Finally, by Lemma~\ref{lem:ccjoin} and $\cc(H_\lambda)=P_\lambda$, the right-hand side is equivalent to the graph $\bigcup_{\lambda\in S}H_\lambda$ being disconnected.
\end{proof}

\begin{theorem}[Size of the Reduction]\label{thm:reduction-size}
The output of the reduction satisfies
\[
  p=kq
\]
and
\[
  n=1+\ell2^k(1+r),
\]
where $r$ is the number of nonedges between distinct color classes and $\ell$ is the prime used by the additive-shift verifier. In particular,
\[
  n=O\!\left(k^2q^2(q2^k+4^k)\right)
  =q^{O(1)}2^{O(k)}.
\]
Each label graph can be realized with at most $n-1$ edges, the total number of edges is $O(np)$, and the reduction can be constructed in polynomial time in the explicit output size.
\end{theorem}

\begin{proof}
By construction,
\[
  |X|=1+\ell2^k,
  \qquad
  |Y|=2+r.
\]
The pointed disjunctive-composition ground set therefore has
\[
  n=1+(|X|-1)(|Y|-1)
   =1+\ell2^k(1+r)
\]
vertices. Since
\[
  r\le\binom{k}{2}q^2
  \qquad\text{and}\qquad
  \ell<2\max\{q,2^k\},
\]
we obtain
\[
  n=O\!\left(k^2q^2(q2^k+4^k)\right).
\]
Each partition $P_\lambda$ can be realized by choosing a spanning tree in every block, using at most $n-1$ edges. Hence, the total number of edges is at most $p(n-1)$. All states, blocks, and realizing forests can be enumerated in time polynomial in the explicit output size. Therefore, the entire reduction is deterministic and runs in polynomial time in its explicit output size.
\end{proof}

\begin{corollary}[Reduction to \GLMC]\label{cor:to-primal}
Let $b=p-k$. Assign label $\lambda$ to every edge of $H_\lambda$. If different label graphs use an edge on the same unordered pair of vertices, retain them as parallel edges with different labels. The resulting \GLMC{} instance is a \Yes-instance if and only if the original balanced \MCC{} instance is a \Yes-instance.
\end{corollary}
\begin{proof}
This follows from Lemma~\ref{lem:dual} and \cref{thm:reduction-correct}.
\end{proof}

\section{From Sparse \texorpdfstring{$3$-SAT}{3-SAT} to the Final Conditional Lower Bound}\label{sec:sat}
To derive a lower bound stated only in terms of the target input parameters $n$ and $p$, we directly construct a balanced \MCC{} instance with parameters in the required range.

Let $F$ be a sparse $3$-CNF formula with $N$ variables and $M\le cN$ clauses, where $c$ is the constant from Proposition~\ref{prop:sparse-sat}. Choose an integer $k\ge2$ and partition the clauses arbitrarily into $k$ groups
\[
  \mathcal C_1,\ldots,\mathcal C_k,
\]
each containing at most $\lceil M/k\rceil$ clauses. Let $X_i$ be the set of variables appearing in the clauses of group $i$. Then
\[
  |X_i|
  \le3\left\lceil\frac{M}{k}\right\rceil
  \le\frac{3cN}{k}+3.
\]

The valid candidates in color class $V_i$ are all local assignments
\[
  \sigma_i:X_i\to\{0,1\}
\]
that satisfy every clause in $\mathcal C_i$. For $i\neq j$, connect $\sigma_i\in V_i$ and $\sigma_j\in V_j$ if and only if they agree on their common variables:
\[
  \sigma_i\restr(X_i\cap X_j)
  =
  \sigma_j\restr(X_i\cap X_j).
\]

\begin{lemma}[Correctness of the Grouping Reduction]\label{lem:sat-mcc}
The formula $F$ is satisfiable if and only if the resulting $k$-partite graph contains a multicolored $k$-clique.
\end{lemma}

\begin{proof}
If $\tau$ is a satisfying assignment of $F$, then $\tau\restr X_i$ is a genuine candidate in $V_i$, and any two such restrictions agree on their common variables. Hence, they form a multicolored clique.

Conversely, suppose that $\sigma_i\in V_i$ form a multicolored clique. Pairwise consistency guarantees that, for every variable appearing in several sets $X_i$, all local assignments give it the same value. Thus, the local assignments can be merged into a single assignment on $\bigcup_iX_i$. This assignment satisfies every clause in every group and therefore satisfies $F$. Variables not appearing in the formula can be assigned arbitrarily.
\end{proof}

\subsection{Balancing and Construction Time}
Choose a constant $C>3c+4$ and let
\[
  q=2^{\lceil CN/k\rceil}.
\]
For all sufficiently large $N$,
\[
  |V_i|
  \le2^{|X_i|}
  \le q.
\]
Add enough isolated dummy vertices to every color class to make its size exactly $q$. When $k\ge2$, no isolated dummy vertex can belong to a multicolored $k$-clique, so balancing preserves equivalence. Therefore,
\[
  |V_i|=q,
  \qquad
  \log q=\Theta(N/k).
\]

Enumerating all local assignments in each group takes $2^{O(N/k)}\poly(N,k)$ time. Checking consistency for all candidate pairs between color classes takes $O(k^2q^2\poly(N))=2^{O(N/k)}\poly(N,k)$ time. Hence, the balanced \MCC{} instance can be constructed explicitly in
\[
  2^{O(N/k)}\poly(N,k)
\]
time.

\subsection{Parameter Choice and Proof of the Lower Bound}\label{sec:parameters}

\subsubsection{Parameter Choice}
For all sufficiently large $N$, let
\[
  k=\lfloor\sqrt N\rfloor.
\]
Then
\[
  \frac{N}{k}=\Theta(k),
  \qquad
  N=\Theta(k^2).
\]
By the balancing step above,
\[
  \log q=\Theta(N/k)=\Theta(k).
\]
The prime $\ell$ used by the selector satisfies
\[
  \max\{q,2^k\}<\ell<2\max\{q,2^k\},
\]
and hence
\[
  \log\ell=\Theta(k).
\]

\subsubsection{Output Size}
By \cref{thm:reduction-size},
\[
  n=1+\ell2^k(1+r),
  \qquad
  r\le\binom{k}{2}q^2.
\]
For the lower bound, $\ell>2^k$ gives
\[
  n\ge1+\ell2^k>4^k,
\]
so
\[
  \log n=\Omega(k).
\]
For the upper bound, $q=2^{\Theta(k)}$, $\ell=2^{\Theta(k)}$, and $r\le k^2q^2$, whence
\[
  n=2^{O(k)}.
\]
Therefore,
\[
  \log n=\Theta(k).
\]
The number of labels $p=kq$ satisfies
\[
  \log p=\Theta(k),
  \qquad
  \log(np)=\Theta(k).
\]
In particular, $p=n^{O(1)}$ on the hard family.

\subsubsection{Conditional Lower Bound for Multigraphs}
We first prove \cref{thm:main} for multigraphs.

\begin{proof}[Proof of \cref{thm:main} for multigraphs]
Suppose that there is a deterministic \GLMC{} algorithm with running time
\[
  (np)^{o(\log n)}\operatorname{poly}(|E|).
\]
For the instances constructed above, we have $\log n=\Theta(k)$. Therefore, on this family of instances, an exponent that is \(o(\log n)\) is also \(o(k)\).

Since every label graph is realized by a forest, the total number of edges is $O(np)$. The binary logarithm of the running time $T$ therefore satisfies
\begin{align*}
  \log T
  &\le o(\log n)\cdot\log(np)+O(\log|E|)\\
  &=o(k)\cdot\Theta(k)+O(k)\\
  &=o(k^2)\\
  &=o(N).
\end{align*}
Thus, $T=2^{o(N)}$.

The balanced \MCC{} instance is constructed in
\[
  2^{O(N/k)}\poly(N,k)=2^{O(k)}\poly(N)
\]
time, and the subsequent reduction runs in time polynomial in the output size, which is $2^{O(k)}$. Hence the complete reduction overhead is $2^{O(k)}\poly(N)=2^{o(N)}$. We would therefore obtain a deterministic $2^{o(N)}$-time algorithm for sparse $3$-SAT, contradicting Proposition~\ref{prop:sparse-sat}.
\end{proof}

\subsection{Transferring the Lower Bound from Multigraphs to Simple Graphs}
\label{subsec:simple-graph-reduction}

Different label graphs in our construction may use edges on the same pair of vertices with different labels, so the reduction naturally produces a multigraph. We now give a polynomial-size, optimum-preserving transformation to simple graphs, thereby transferring the lower bound from multigraphs to simple graphs.

Let
\[
I=(G,\chi,b)
\]
be a \textsc{Global Label Min-Cut} instance, where $G=(V,E)$ is an undirected edge-labeled multigraph,
\[
\chi:E\to[p]
\]
is the labeling function, and $b\leq p$ is the budget. Instances with $|V|\leq1$ are degenerate and can be decided directly, so throughout this subsection we assume that $|V|\geq2$.

Let
\[
n=|V|, \qquad m=|E|,
\]
where parallel edges are counted separately. Self-loops never cross a cut and can therefore be deleted in advance. Hence, after this preprocessing, every edge $e=u_ev_e$ satisfies $u_e\neq v_e$.

\begin{lemma}[Simple-Graph Transformation]
\label{lem:multigraph-to-simple}
Let $I=(G,\chi,b)$ be a loopless instance with $|V|\geq2$ and $b\leq p$. In polynomial time, one can construct an instance on a simple edge-labeled graph
\[
I^\sharp=(G^\sharp,\chi^\sharp,b)
\]
such that
\[
I\text{ is a \Yes-instance}\quad\Longleftrightarrow\quad I^\sharp\text{ is a \Yes-instance}.
\]
Moreover,
\[
\operatorname{OPT}(G^\sharp,\chi^\sharp)=\operatorname{OPT}(G,\chi).
\]
\end{lemma}

\begin{proof}
For every original edge
\[
e=u_ev_e\in E,\qquad \chi(e)=\lambda_e,
\]
arbitrarily designate one endpoint as $u_e$ and the other as $v_e$. Introduce new vertices
\[
x_e,w_{e,1},\ldots,w_{e,p},
\]
and let
\[
C_e=\{u_e,x_e,w_{e,1},\ldots,w_{e,p}\}.
\]
Thus,
\[
|C_e|=p+2.
\]

Add a complete graph $K_{p+2}$ on $C_e$. Assign a protection label to every edge of this complete graph, subject to the following requirements:
\begin{enumerate}
  \item the protection labels within the same $C_e$ are pairwise distinct;
  \item protection labels used for different edge gadgets are pairwise distinct;
  \item every protection label is distinct from every original label.
\end{enumerate}
Finally, add the edge
\[
x_ev_e
\]
with the original label $\lambda_e$, and delete the original edge $u_ev_e$.

Every vertex among $x_e,w_{e,1},\ldots,w_{e,p}$ belongs only to the gadget of edge $e$, so distinct gadgets cannot create edges with the same pair of endpoints. Moreover, $x_e\notin V$ and $u_e\neq v_e$, so the edge $x_ev_e$ cannot coincide with an edge of the protection clique. Therefore, $G^\sharp$ is simple.

First, let
\[
S\mid V\setminus S
\]
be an arbitrary nontrivial cut of the original graph. For every original edge $e=u_ev_e$, place
\[
x_e,w_{e,1},\ldots,w_{e,p}
\]
on the same side as $u_e$. Then the entire protection clique $C_e$ lies on one side of the cut, so no protection label crosses the cut. The only edge of this gadget that can cross is
\[
x_ev_e.
\]
Since $x_e$ and $u_e$ lie on the same side,
\[
x_ev_e\text{ crosses the cut} \quad\Longleftrightarrow\quad u_ev_e\text{ crosses the original cut}.
\]
The two edges carry the same original label $\lambda_e$. Hence, the extended cut crosses exactly the same set of labels as the original cut, and therefore
\[
\operatorname{OPT}(G^\sharp,\chi^\sharp)\le\operatorname{OPT}(G,\chi).
\]

Conversely, consider any nontrivial cut of $G^\sharp$ of cost at most $p$. Every nontrivial cut of $K_{p+2}$ crosses at least $p+1$ edges: if one side contains $r$ vertices, then the number of crossing edges is
\[
r(p+2-r)\ge p+1.
\]
Since every edge of a protection clique has a distinct protection label, any cut splitting some $C_e$ crosses at least $p+1$ distinct labels. Therefore, no cut of cost at most $p$ can split any protection clique $C_e$.

Consequently,
\[
u_e,x_e,w_{e,1},\ldots,w_{e,p}
\]
all lie on the same side of such a cut, and hence
\[
x_ev_e\text{ crosses the cut}\quad\Longleftrightarrow\quad u_ev_e\text{ crosses the induced cut on the original vertex set }V.
\]

The induced cut on $V$ must be nontrivial. Otherwise, if all original vertices lie on the same side, then every protection clique, which cannot be split, must lie entirely on that side as well. All new vertices would then lie on the same side, contradicting the assumption that the cut of $G^\sharp$ is nontrivial.

Thus, every simple-graph cut of cost at most $p$ induces a nontrivial cut in the original graph, and the two cuts cross exactly the same set of original labels. Since every cut in the original graph crosses at most all $p$ original labels,
\[
\operatorname{OPT}(G,\chi)\le p.
\]
The first direction therefore implies that $G^\sharp$ has a cut of cost at most $p$, so an optimum cut in $G^\sharp$ cannot split a protection clique. We conclude that
\[
\operatorname{OPT}(G^\sharp,\chi^\sharp)=\operatorname{OPT}(G,\chi).
\]

In particular, for every $b\le p$,
\[
G\text{ has a label cut of cost at most }b\quad\Longleftrightarrow\quad G^\sharp\text{ has a label cut of cost at most }b.
\]
\end{proof}

Each original edge introduces $p+1$ vertices, and hence
\[
n^\sharp=n+m(p+1).
\]
The protection clique $C_e$ contains
\[
\binom{p+2}{2}
\]
edges, each carrying a fresh protection label. Therefore,
\[
p^\sharp=p+m\binom{p+2}{2}.
\]
The number of edges is
\[
m^\sharp=m+m\binom{p+2}{2}.
\]
Consequently,
\[
n^\sharp,p^\sharp,m^\sharp\le (n+m+p)^{O(1)}.
\]
This is a polynomial-time, optimum-preserving reduction.

\begin{corollary}[Conditional Lower Bound for Simple Graphs]
\label{cor:simple-graph-hardness}
Suppose that a hard family of multigraph instances satisfies
\[
  m\le(np)^{O(1)}
\]
and $p\le n^{O(1)}$. If \GLMC{} on edge-labeled simple graphs admitted a deterministic algorithm with running time
\[
  (np)^{o(\log n)}\operatorname{poly}(|E|),
\]
then this would yield an algorithm in the same asymptotic running-time class on the original multigraph family.
\end{corollary}

\begin{proof}
By Lemma~\ref{lem:multigraph-to-simple}, the transformation preserves the \Yes/\No{} answer and satisfies
\[
  n^\sharp p^\sharp=(np)^{O(1)}.
\]
Moreover, $n^\sharp\ge n$, while the assumptions $m\le(np)^{O(1)}$ and $p\le n^{O(1)}$ imply
\[
  n^\sharp,p^\sharp\le n^{O(1)}.
\]
Consequently,
\[
  \log n^\sharp=\Theta(\log n),
  \qquad
  \log(n^\sharp p^\sharp)=\Theta(\log n)=\Theta(\log(np)).
\]
Since
\[
  m^\sharp\le (np)^{O(1)}
\]
on the hard family, an algorithm with running time
\[
  (n^\sharp p^\sharp)^{o(\log n^\sharp)}\operatorname{poly}(m^\sharp)
\]
would satisfy
\begin{align*}
  \log T^\sharp
  &\le o(\log n^\sharp)\log(n^\sharp p^\sharp)+O(\log m^\sharp)\\
  &=o(\log n)\cdot\Theta(\log(np))+O(\log(np))\\
  &=o(\log n)\log(np).
\end{align*}
Consequently, the transformed algorithm would run in
\[
  (np)^{o(\log n)}
\]
time, contradicting the multigraph lower bound.
\end{proof}

For the hard family constructed in this paper, \cref{thm:reduction-size} gives $m=O(np)$, while \cref{sec:parameters} gives $\log p=\Theta(\log n)$ and hence $p\le n^{O(1)}$. Combining the multigraph lower bound with Corollary~\ref{cor:simple-graph-hardness} completes the proof of \cref{thm:main}.

\begin{proof}[Proof of Corollary~\ref{cor:randomized-lower-bound}]
For the randomized lower bound, apply the construction of \cref{sec:sat} with the sparsity constant $c_{\mathrm r}$ from Proposition~~\ref{prop:sparse-sat-rETH} in place of $c$, and choose $C>3c_{\mathrm r}+4$. All reductions used in the proof of \cref{thm:main} are deterministic. Hence, a bounded-error randomized algorithm for
Global Label Min-Cut with running time
\[
  (np)^{o(\log n)}\operatorname{poly}(|E|)
\]
would, when composed with these reductions, yield a bounded-error randomized $2^{o(N)}$-time algorithm for sparse 3-SAT. This contradicts Proposition~\ref{prop:sparse-sat-rETH}.
\end{proof}

\section{Conclusion}\label{sec:conclusion}
We give a deterministic fine-grained reduction from sparse $3$-SAT to \GLMC{} that eliminates both $\log\log n$ losses in the previously known conditional lower bound. The proof separates into three layers: exact-one selection validity, clique consistency, and logical composition. The nonedge verifier handles clique consistency only under the exact-one promise, while the pointed partition disjunctive composition combines the two rejection conditions without losing the join structure.

Under ETH, we conclude that \GLMC{} admits no deterministic algorithm with running time
\[
    (np)^{o(\log n)}\operatorname{poly}(|E|),
\]
even on simple edge-labeled graphs. Since all reductions are deterministic, the same lower bound applies to bounded-error randomized algorithms under rETH. On the hard family, \(p=n^{O(1)}\); hence, in the randomized setting, the result matches the worst-case exponent order of the known randomized quasi-polynomial exact upper bound, up to constant factors.

\end{document}